\PassOptionsToPackage{usenames,dvipsnames}{xcolor}
\documentclass[11pt]{article}

\def\ShowComments{true} 

\def\AlgoStyle{algorithm2e} 

\usepackage[margin=1in]{geometry}
\usepackage[T1]{fontenc}
\usepackage[utf8]{inputenc}
\usepackage{charter}
\usepackage[english]{babel}
\usepackage{setspace}
\usepackage{multicol}
\usepackage[inline]{enumitem}
\usepackage{csquotes}

\usepackage{amsmath, amssymb, amstext, amsthm}
\usepackage{mathtools}
\usepackage{mathrsfs}
\usepackage{nicefrac}
\usepackage{mleftright} 
\usepackage{thmtools, thm-restate}
\DeclareMathAlphabet{\mathbbold}{U}{bbold}{m}{n}

\usepackage{graphicx}
\usepackage[dvipsnames, svgnames, x11names]{xcolor}
\usepackage{tcolorbox}
\usepackage{soul} 

\usepackage{ifthen} 
\usepackage{hyperref}
\usepackage{cleveref} 
\usepackage{refcount}

\ifthenelse{\equal{\AlgoStyle}{algorithm2e}}{
    \usepackage[noend,ruled,linesnumbered]{algorithm2e}
    
    \SetAlFnt{\small} \SetAlCapFnt{\small} \SetAlCapNameFnt{\small}
    \SetAlCapHSkip{0pt} \IncMargin{-\parindent}
    \crefname{algocf}{alg.}{algs.}
    \Crefname{algocf}{Algorithm}{Algorithms}
}{}
\ifthenelse{\equal{\AlgoStyle}{algorithmic}}{
    \usepackage{algorithm}
    \usepackage{algpseudocode}

}{}

\usepackage{todonotes}
\newcommand{\hide}[1]{} 

\ifthenelse{\equal{\ShowComments}{true}}{

}{
    \presetkeys{todonotes}{disable}{}

}

\declaretheorem[numberwithin=section]{theorem}

\declaretheorem[numberlike=theorem]{claim,observation,corollary,proposition,lemma}

\Crefname{assumption}{Assumption}{Assumptions}

\declaretheorem[numberwithin=theorem,name=Claim]{claim-inside-theorem}
\Crefname{claim-inside-theorem}{Claim}{Claims}
\declaretheorem[numberlike=claim-inside-theorem,name=Lemma]{lemma-inside-theorem}
\Crefname{lemma-inside-theorem}{Lemma}{Lemmas}
\declaretheorem[numberlike=claim-inside-theorem,name=Observation]{observation-inside-theorem}

\declaretheorem[numberwithin=lemma,name=Claim]{claim-inside-lemma}
\Crefname{claim-inside-lemma}{Claim}{Claims}
\declaretheorem[numberlike=claim-inside-lemma,name=Corollary]{corollary-inside-lemma}
\declaretheorem[numberwithin=lemma,name=Lemma]{lemma-inside-lemma}
\declaretheorem[numberlike=claim-inside-lemma,name=Observation]{observation-inside-lemma}

\declaretheorem[numberwithin=lemma-inside-theorem,name=Claim]{claim-inside-lemma-inside-theorem}
\declaretheorem[numberwithin=lemma-inside-theorem,name=Lemma]{lemma-inside-lemma-inside-theorem}
\declaretheorem[numberwithin=lemma-inside-theorem,name=Observation]{observation-inside-lemma-inside-theorem}

\newcommand{\etal}{et al.\xspace}

\newcommand{\bigoh}{\mathcal{O}}

\newcommand{\W}{\ensuremath{{\rm {\sf W}}}}

\newcommand{\instI}{\mathcal{I}\xspace}
\newcommand{\instJ}{\mathcal{J}\xspace}

\addto\extrasenglish{}

\newcommand\dsfull{\textsc{Dominating Set}\xspace}
\newcommand\ds{\textsc{DS}\xspace}
\newcommand\dsqfull{\textsc{Dominating Set with Quotas}\xspace}
\newcommand\dsq{\textsc{DSQ}\xspace}

\newcommand\scqfull{\textsc{Set Cover with Quotas}\xspace}
\newcommand\scq{\textsc{SCQ}\xspace}

\newcommand\is{\textsc{IS}\xspace}

\newcommand\ecfull{\textsc{Exact-Equitable Coloring}\xspace}
\newcommand\ec{\textsc{E-EC}\xspace}
\newcommand\floq{\mathrm{f_{lq}}\xspace} 
\newcommand\fupq{\mathrm{f_{uq}}\xspace} 
\newcommand\tw{\mathrm{tw}}
\newcommand\gd{\textsc{Generalized Domination} \xspace}

\newcommand\F{\mathcal{F}\xspace} 
\newcommand\U{\mathcal{U}\xspace} 

\newcommand\C{\mathcal{C}\xspace} 

\newtheorem{rr}{Reduction Rule}
\crefname{rr}{Reduction Rule}{Reduction Rules}

\newcommand{\defparprob}[4]{
\begin{tcolorbox}[colback=purple!5!white,colframe=black!75!black]
  \vspace{-1mm}
  \begin{tabular*}{\textwidth}{@{\extracolsep{\fill}}lr} #1  \\  \end{tabular*}
  {\bf{Input:}} #2  \\
  {\bf{Parameter:}} #3 \\
  {\bf{Question:}} #4
  \vspace{-1mm}
\end{tcolorbox}
}
\newcommand{\fpt}{\ensuremath{{\rm {\sf FPT}}}\xspace}

\usepackage{booktabs,nicematrix}
\usepackage{multirow}

\begin{document}
\title{Dominating Set with Quotas: Balancing Coverage and Constraints}
\author{
    Sobyasachi Chatterjee
    \thanks{The Institute of Mathematical Sciences, HBNI, Chennai, India. {\tt sobyasachic@imsc.res.in}}
    \and
    Sushmita Gupta 
    \thanks{The Institute of Mathematical Sciences, HBNI, Chennai, India. {\tt sushmitagupta@imsc.res.in}} 
    \and 
    Saket Saurabh
    \thanks{The Institute of Mathematical Sciences, HBNI, Chennai, India. {\tt saket@imsc.res.in}}~
    \thanks{University of Bergen, Norway.}
    \and 
    Sanjay Seetharaman
    \thanks{The Institute of Mathematical Sciences, HBNI, Chennai, India. {\tt sanjays@imsc.res.in}} 
    \and
    Anannya Upasana
    \thanks{The Institute of Mathematical Sciences, HBNI, Chennai, India. {\tt anannyaupas@imsc.res.in}}
}
\date{}

\maketitle

\begin{abstract}
We study a natural generalization of the classical \textsc{Dominating Set} problem, called \textsc{Dominating Set with Quotas} (\dsq). In this problem, we are given a graph \( G \), an integer \( k \), and for each vertex \( v \in V(G) \), a lower quota \( \mathrm{lo}_v \) and an upper quota \( \mathrm{up}_v \). The goal is to determine whether there exists a set \( S \subseteq V(G) \) of size at most \( k \) such that for every vertex \( v \in V(G) \), the number of vertices in its closed neighborhood that belong to \( S \), i.e., \( |N[v] \cap S| \), lies within the range \( [\mathrm{lo}_v, \mathrm{up}_v] \). This richer model captures a variety of practical settings where both under- and over-coverage must be avoided—such as in fault-tolerant infrastructure, load-balanced facility placement, or constrained communication networks.

While \ds{} is already known to be computationally hard, we show that the added expressiveness of per-vertex quotas in \dsq{} introduces additional algorithmic challenges. In particular, we prove that \dsq{} becomes \W[1]-hard even on structurally sparse graphs—such as those with degeneracy 2, or excluding \( K_{3,3} \) as a subgraph—despite these classes admitting FPT algorithms for \ds{}. On the positive side, we show that \dsq{} is fixed-parameter tractable when parameterized by solution size and treewidth, and more generally, on nowhere dense graph classes. Furthermore, we design a subexponential-time algorithm for \dsq{} on apex-minor-free graphs using the bidimensionality framework. These results collectively offer a refined view of the algorithmic landscape of \dsq{}, revealing a sharp contrast with the classical \ds{} problem and identifying the key structural properties that govern tractability. 

\end{abstract}
\newpage
\section{Introduction}

{\sc Dominating Set} (\ds, in short) is a prototypical covering problem studied on graphs. In this problem, we are given a graph $G$ and an integer $k$, and the goal is to find a subset of vertices $S$ of size at most $k$ such that for each vertex $v \in V(G)$, either $v\in S$ or a neighbor of $v$ is in $S$. This is one of the classic \textsf{NP}-complete problems, listed in Garey and Johnson~\cite{garey1979computers}, and it remains intractable even on planar graphs of maximum degree 3. It could be viewed as a graph version of the {\sc Set Cover} problem. Since the problem is intractable, it has been studied extensively using algorithmic approaches meant for coping with \textsf{NP}-hardness, including approximation~\cite{johnson1973approximation}, exact~\cite{iwata2012faster}, and parameterized algorithms~\cite{cygan2015parameterized}.

Together with {\sc Vertex Cover}, the {\sc Dominating set} problem can be considered the Drosophila of parameterized complexity. These problems provide a fertile ground for testing new ideas and tools. This has led to the study of several variants and generalizations of \ds. This includes adding constraints on solutions (such as independence or connectivity) or constraints on how vertices outside the solution interact with the solution vertices. These led to problems such as {\sc Connected Dominating Set}, {\sc Independent Dominating Set}, {\sc $[1,j]$-Total Dominating Set}, and {\sc $(\sigma,\rho)$-Dominating Set}~\cite{DBLP:conf/esa/GuhaK96,DBLP:conf/fsttcs/GuhaK98,meybodi2020parameterized,telle1994complexity}, to name a few.

In this paper, we initiate a systematic study of another variant of {\sc Dominating Set}, called \dsqfull: a generalization of \ds in which each vertex has a lower and an upper-quota on domination, motivated by applications in resource/fair allocation~\cite{banerjee2023allocating}. This richer model captures situations where having too little or too much coverage is undesirable. Such cases often appear in real-world systems that need to manage redundancy, balance load, or limit communication costs. Formally, the problem is stated as follows.

\vspace{-0.2cm}
  \defparprob
{\dsqfull (\dsq, in short)}
{An undirected graph $G$, an integer $k$, and domination 
  quotas $\floq,\fupq \colon V(G)\rightarrow \mathbb{N}\cup\{ 0\}$. }
{ $k$}
  {A vertex $v$ is said to be dominated \textit{properly} by a set $S \subseteq V(G)$ if $\floq(v) \le \vert N[v]\cap S\vert \le \fupq(v)$. Is there a subset of vertices $S$ such that $|S| \le k$ and $S$ properly dominates $V(G)$?}
\vspace{-0.1cm}


In this paper, we study \dsq{} in the framework of parameterized complexity. It is well known that \ds{} is \W[2]-complete and, moreover, hard to approximate even within \fpt{} time~\cite{DBLP:journals/siamcomp/DowneyF95,karthik2018parameterized}. Since \dsq{} is a generalization of \ds{}, these hardness results naturally extend to \dsq{} as well. Therefore, to obtain positive algorithmic results for \dsq{}, we must restrict our attention to graph classes where \ds{} is known to be fixed-parameter tractable (\fpt)—that is, where it admits an algorithm with running time of the form \( f(k) \cdot n^{\bigoh(1)} \), with \( n \) denoting the number of vertices in the input graph.

Although \ds{} is \W[2]-complete on general graphs, it is known to be fixed-parameter tractable on several sparse graph classes, including planar graphs, graphs of bounded genus, graphs excluding a fixed graph \( H \) as a (topological) minor, graphs of bounded expansion, nowhere dense graphs, and biclique-free graphs. Notably, biclique-free graphs subsume all of these classes and form the largest known family of graphs on which \ds{} admits an \fpt algorithm. 

Research on \ds{} in sparse graph classes has been one of the most fruitful directions in parameterized complexity, yielding several powerful tools and techniques~\cite{DBLP:journals/algorithmica/AlonG09,DBLP:journals/cj/DemaineH08,DBLP:conf/stacs/DrangeDFKLPPRVS16,DBLP:conf/iwpec/EinarsonR20,DBLP:conf/stacs/FabianskiPST19,DBLP:journals/siamcomp/FominT06,DBLP:conf/soda/FominLST12,DBLP:conf/stacs/FominLST13,DBLP:journals/siamcomp/FominLST20,DBLP:journals/jcss/GajarskyHOORRVS17,DBLP:journals/tcs/LokshtanovMS11}. This line of work is not only theoretically rich but also practically motivated. For many real-world applications, it is reasonable to assume that the solution size parameter \( k \) is small and that the input graphs are structurally sparse—such as having bounded degree, low degeneracy, or constrained topological features. These sparse classes capture a broad range of practical settings. For instance, planar graphs naturally model geographic networks or physical layouts, making them highly relevant in applied contexts.

Motivated by this, our main objective in this paper is to investigate the parameterized complexity of \dsq{} by drawing upon the extensive work on \ds{} and its behavior on structurally sparse graph classes. Our goal is to chart the boundary between tractable and intractable cases for \dsq{} across this rich landscape.

{
\begin{table}[t]
    \small
    \centering
    \begin{NiceTabular}{|c|p{30mm}|p{70mm}|}
        \toprule
        \bf Graph class & { \bf \ds }& \bf \dsq \\
        \toprule
        Bounded degeneracy & { $\fpt(k,d)$ \cite{DBLP:journals/algorithmica/AlonG09}} & \W-hard on $2$-degenerate graphs (Thm. \ref{thm:dsq-hard-3dgen})\\
        \midrule
        \multirow{2}{10em}{\centering Bounded treewidth}  & \multirow{2}{10em}  {$\fpt(\tw)$ \cite{arnborg1989linear} } & \W-hard par. $\tw$ (Thm. \ref{thm:dsq-hard-tw})\\
        & & $\fpt(k, \tw)$ (Thm. \ref{thm:dsq-fpt-ktw})\\
        \midrule
        Apex-minor-free    & { $2^{\bigoh(\sqrt{k})}$ } \cite{fomin2009contraction} & $2^{\bigoh(\sqrt{k} \log k)}$ (Thm. \ref{thm:dsq-fpt-amf})\\
        \midrule
        Nowhere dense   & { $\fpt(k)$ \cite{dawar4domination} } & $\fpt(k)$ (Thm. \ref{thm:dsq-fpt-nd})\\
        \midrule
        Biclique-free      & { $\fpt(k,t)$ \cite{telle2012fpt} }& \W-hard on $K_{3,3}$-free graphs (from Thm. \ref{thm:dsq-hard-3dgen})\\
        \bottomrule
    \end{NiceTabular}
    \caption{An overview of the \fpt status of \ds and \dsq in various graph classes.}
    \label{tab:our_results}
\end{table}
}

\subsection{Our results}\label{subsection:results}

We start by showing that \dsq parameterized by solution size is \W[1]-hard even on graphs of degeneracy $2$, \Cref{thm:dsq-hard-3dgen}. A graph $G$ has degeneracy $d$ if every subgraph of $G$ has a vertex of degree at most $d$. This is in sharp contrast to \ds which has been shown to be \fpt parameterized by $k+d$ on graphs of degeneracy $d$ by Alon and Gutner~\cite{DBLP:journals/algorithmica/AlonG09}. Moreover, \ds is \fpt parameterized by $k+t$ on $K_{t,t}$-free graphs (i.e., graphs that exclude the biclique $K_{t,t}$ as a subgraph)~\cite{telle2012fpt}. Since graphs with degeneracy $d$ are $K_{d+1,d+1}$-free, it follows that \dsq is \W[1]-hard even on $K_{3,3}$-free graphs. Our next result shows that \dsq parameterized by treewidth $(\tw)$ is \W[1]-hard. This is unlike \ds which is known to admit an algorithm with running time $3^\tw n^{\bigoh(1)}$ on graphs of treewidth $\tw$~\cite{DBLP:conf/esa/RooijBR09,lokshtanov2018known}, that is tight under SETH~\cite{DBLP:journals/jcss/ImpagliazzoP01}. We complement this intractability result by showing that \dsq can be solved in time $2^\tw (k+1)^{2\tw} \cdot n^{\bigoh(1)}$ on graphs of treewidth $\tw$ using standard dynamic programming on tree decompositions. Using a ``subset convolution''-like method, we improve this to a $((2k+2)^{\tw} +2^{\bigoh(k)})\cdot n^{\bigoh(1)}$ time algorithm. For this, we follow the underlying proof of the subset convolution method given in \cite{van2021generic} rather than using the tool in a black-box fashion.

Using the \textit{bidimensionality} approach of Fomin~\etal~\cite{fomin2009contraction}, which works based on the structure of graphs of large treewidth, in conjunction with our algorithm on graphs of bounded treewidth, we show that \dsq can be solved in time $2^{\bigoh(\sqrt{k} \log k)} \cdot n^{\bigoh(1)}$ on apex-minor-free graphs. The largest graph class for which we show that \dsq is \fpt is the class of nowhere dense graphs, which includes several graph classes such as (topological) minor-free graphs and graphs of bounded expansion. We prove this result by expressing \dsq using First-Order (FO) logic and then invoking the result of Grohe~\etal~\cite{grohe2017deciding} that FO model-checking is \fpt in the size of the formula on nowhere dense graphs.

In \Cref{sec:appendix-sc}, we study \scqfull, a generalization of \dsq, and show that it is $\fpt$ parameterized by $k+d$ when the sets are of size at most $d$. 
While the same is easy to see for \textsc{Set Cover} (for e.g., by dynamic programming), the presence of elements with lower-quota zero makes the result interesting. Since \dsq on graphs of maximum degree $d$ is a special case of \scq with sets of bounded size, we also obtain an algorithm for \dsq on graphs of bounded degree. 

\subsection{Related work}

The {\sc Threshold Dominating Set} problem is a special case of \dsq{}, where we are given a graph \( G \), positive integers \( k \) and \( r \), and the objective is to find a set \( S \subseteq V(G) \) such that \( |S| \leq k \) and for every vertex \( v \in V(G) \), it holds that \( |N[v] \cap S| \geq r \). Golovach and Villanger~\cite{golovach2008parameterized} showed that this problem is fixed-parameter tractable on graphs of bounded degeneracy. A related problem, studied by Meybodi~\etal~\cite{meybodi2020parameterized}, is the \([1,j]\)-{\sc Total Dominating Set}: given a graph \( G \) and an integer \( k \), the goal is to find a set \( S \subseteq V(G) \) of size at most \( k \) such that for each vertex \( v \in V(G) \), it holds that \( 1 \leq |N(v) \cap S| \leq j \), for some fixed \( j \in \mathbb{N} \). 
They show that \([1,j]\)-{\sc Total Dominating Set} can be solved in time $(2j + 2)^\tw \cdot n^{\bigoh(1)}$.
They also study the $[1,j]$-\textsc{Dominating Set} problem where given a graph $G$ and an integer $k$, the goal is to find a set $S$ of vertices of size at most $k$ such that for each $v \in V(G) \setminus S$, $1 \le |N(v) \cap S| \le j$ for some fixed $j \in \mathbb{N}$.
They show that $[1,j]-$\textsc{Dominating Set} is \W[1]-hard in $(j+1)$-degenerate graphs; can be solved in time $(j + 2)^\tw \cdot n^{\bigoh(1)}$; and is \fpt w.r.t. $k$ in nowhere dense graphs using FO model-checking.
Another closely related problem is the {\sc [$\sigma, \rho$]-Dominating Set}, also known as {\sc Generalized Domination}, introduced by Telle~\cite{telle1994complexity}. In this model, the quota conditions for each vertex depend on whether the vertex is included in the solution. This contrasts with \dsq{}, where each vertex has a fixed quota that is independent of whether it belongs to the solution, and different vertices may have distinct lower and upper bounds. \gd has been well-studied (see \cite{golovach2012parameterized}) and captures problems like \ds, \textsc{Independent Set}, \textsc{Perfect Code}, and \textsc{Strong Stable Set}, for example. Formally, given a graph $G$, and sets $\sigma, \rho$ of non-negative integers, the goal is to find a subset of vertices, $S$, such that $|N(v) \cap S| \in \sigma$ for each $v \in S$ and $|N(v) \cap S| \in \rho$ for each $v \in V(G)\setminus S$. Note that even deciding the existence of a solution is NP-hard in some cases (\textsc{Perfect Code}, for example \cite{cull1998perfect}). The existence, the optimization, and the counting versions of \gd have also been studied, \cite{DBLP:conf/esa/RooijBR09,van2021generic}. We note that the sets $\sigma$ and $\rho$ can be cofinite.

Observe that $[1,j]$-\textsc{Total Dominating Set} is equivalent to \gd with $\sigma = [j-1]$ and $\rho = [j]$.
Assuming $j$ is a constant, van Rooij \cite{van2021generic} presented an algorithm that solves this case of \gd in time $(2j +1)^\tw n^{\bigoh(1)}$ and Focke et al. \cite{focke2023tight} show a matching SETH-based lowerbound.
An $n^{f(\tw)}$ algorithm for \gd was shown by Knop \etal \cite{knop2019simplified} using a model-checking approach.
They study the logic $\textsf{MSO}^\textsf{GL}$ (an extension of $\textsf{MSO}_2$ logic with global and linear constraints: these constraints allow expressing the values in $\sigma$ and $\rho$) and show that model-checking in that logic can be done in time $n^{f(\tw(G),||\phi||)}$ for some computable function $f$.
Another related problem that has been studied in \cite{knop2019parameterized,masavrik2020parameterized} is that of \textsc{Fair $\textup{L}$ Vertex Evaluation}: given a logic $\textup{L}$ formula $\phi$, an undirected graph $G$, and an integer $k$, is there a set $W \subseteq V(G)$ such that $G \models \phi(W)$ and for every vertex $v$ of $G$, it holds that $|N(v) \cap W| \le k$. Note that the last constraint captures the fairness of the set $W$ with respect to each vertex and corresponds to each vertex having between $0$ and $k$ neighbors in $W$.

Gupta~\etal~\cite{DBLP:conf/sofsem/GuptaJPS21} studied the complexity of \textsc{$d$-Hitting Set with Quotas}, a problem equivalent to \textsc{Set Cover with Quotas} where each element appears in at most $d$ sets. We note that neither \dsq (where the degree is unbounded) nor \textsc{Set Cover with Quotas} subsumes the other. The literature on variants of \dsfull is vast; to the best of our knowledge, the problems mentioned above are the most closely related to ours.

\subsection{Preliminaries}
\hide{We use \dsq as shorthand for both the problem name and a solution to it.}
We say that a vertex $v$ has quota $\langle i,j \rangle$ if $\floq(v) = i$ and $\fupq(v) = j$.
\hide{We will assume that the upper-quota of all vertices is at most $k$ since we are looking for solutions of size at most $k$.}
For integers $i$ and $j$, we denote the sets $\{i, \dots, j\}$ and $\{1, \dots, j\}$ by $[i,j]$ and $[j]$, resp.
We use $n$ and $m$ to denote the number of vertices and edges in the input graph, respectively.
\hide{A vertex $v$ is a neighbor of vertex $u$ if there is an edge $vu \in E(G)$.
By $N(v)$ we denote the neighborhood (set of neighbors) of vertex $v$ and by $N[v]$ we denote the closed neighborhood of $v$: $N[v] = N(v) \cup \{v\}$.}
We use $N^d_G(v)$ to denote the set of vertices at distance at most $d$ from $v$ in $G$.

\begin{sloppypar}
\paragraph{Parameterized Complexity.} A parameterized problem (for e.g., \ds parameterized by solution size) is said to be \textit{fixed-parameter tractable} (\fpt or $\fpt$(parameter), in short) if an instance $(\instI,k)$ of the problem can be solved in time $f(k) \cdot |\instI|^{\bigoh(1)}$ for some computable function $f$ \cite{cygan2015parameterized}.
The set of fixed-parameter tractable problems form the complexity class \fpt.
The \W-hierarchy is a set of complexity classes defined to finely capture the hardness of parameterized problems: \fpt $\subseteq \W[1] \subseteq \W[2] \subseteq \dots$.
It is believed that \W-hard problems are not in \fpt.

\end{sloppypar}

\paragraph{Minors.} For an edge $e=uv \in E(G)$, the result of \textit{contracting} $e$ in $G$ (denoted by $G/e$) is a graph in which the endpoints $u$ and $v$ are deleted and replaced by a new vertex that is adjacent to all vertices in $V(G)\setminus \{u,v\}$ that either $u$ or $v$ was adjacent to.
A graph $H$ is a \textit{contraction} of $G$ if it can be obtained from $G$ by a sequence of edge contractions.
A graph $H$ is a \textit{minor} of $G$ (denoted by $H \preceq G$) if it can be obtained from $G$ by a sequence of edge contractions, edge deletions, and vertex deletions.
A graph $G$ is $H$-minor-free if $H \npreceq G$.

\begin{sloppypar}
\paragraph{Treewidth.} The treewidth of a graph is a parameter that measures the ``closeness'' of a graph to a tree.
It is defined based on a structural decomposition of a graph called the \textit{tree decomposition}.
Given a graph $G$, a pair $\mathcal{T} = (T,\{X_t\}_{t \in V(T)})$, where $T$ is a tree and $X_t \subseteq V(T)$ is a \textit{bag} of vertices mapped to node $t$ (vertices in $T$ are referred to as nodes), is a tree decomposition of $G$ if it satisfies the following conditions:
  (1) Each edge in $G$ is contained in some bag; that is, $uv \in E(G) \!\implies\! \exists ~ t \in V(T)\colon \{u,v\} \subseteq X_t$;
  (2) For each vertex $v \in V(G)$, the nodes whose bags contain $v$ (that is, $\{t \in V(T)\colon v \in X_t\}$) form a nonempty connected subtree.
The \textit{width} of the tree decomposition $(T,\{X_t\}_{t \in V(T)})$ is defined to be $\max_{t \in V(T)} |X_t| - 1$, one subtracted from the maximum bag size.
The \textit{treewidth} of $G$, denoted by $\tw(G)$, is defined to be the minimum width among all tree decompositions of $G$.
When the context is clear, we write $\tw$ instead of $\tw(G)$.
\end{sloppypar}

\section{\dsq on graphs of bounded treewidth}

In this section, we study \dsq on graphs of bounded treewidth.
Firstly, we show that \dsq is hard when parameterized by $\tw$.
Following that, we show that \dsq is \fpt when parameterized by $\tw + k$.

\subsection{Hardness of DSQ parameterized by treewidth}
\label{sec:dsq-par-tw}

We establish hardness
by a reduction from \ecfull (\ec): given a graph $H$ and an integer $r$, does $H$ admit a proper vertex coloring with exactly $r$ colors such that the sizes of any two color classes differ by at most 1?
The problem is \W[1]-hard parameterized by $\tw(H)\!+r$\cite{fellows2011complexity}.
In \cite{fellows2011complexity}, the authors define \textsc{Equitable Coloring} (EC) as the problem in which an equitable coloring of $H$ using \textit{at most} $r$ colors is found instead of exactly $r$ colors.
However, they present a reduction that establishes the hardness of \ec, which also establishes the hardness of \textsc{EC}.
Next, we observe a property of the color class sizes in an equitable coloring, which forms a key component in our reduction.
\begin{proposition}
  \label{prop:size-eec}
  Given an equitable coloring in a graph with $n$ vertices using $r$ colors, each color class is of size either $\lfloor n/r \rfloor$ or $\lceil n/r \rceil$.
\end{proposition}
\begin{proof}
 If there is a class of size at most $\lfloor n/r \rfloor - 1$, then the maximum number of vertices that can be colored using $r$ colors respecting the equitability condition is $\lfloor n/r \rfloor - 1 + (r-1)\lfloor n/r \rfloor  = r(\lfloor n/r \rfloor)-1 < n$.
 Thus, each class is of size at least $\lfloor n/r \rfloor$.
 
 If there is a class of size at least $\lceil n/r \rceil + 1$, then the minimum number of vertices that can be colored using $r$ colors respecting the equitability condition is $\lceil n/r \rceil + 1 + (r-1)\lceil n/r \rceil  = r(\lceil n/r \rceil)+1 > n$.
 Thus, each class is of size at most $\lceil n/r \rceil$.
\end{proof}

By \Cref{prop:size-eec}, we know the smallest and largest sizes of the color classes. 
The basic idea in the reduction is to express this fact (along with expressing proper vertex coloring) using domination lower and upper-quotas.
Given an instance $\instI = (H,r)$ of \ec, we construct an instance $\instJ = (G,k,\floq,\fupq)$ of \dsq in polynomial time as given in \Cref{reduction2}.
See \Cref{fig:dsq-w-hardness-reduction-tw} for an illustration.

\begin{algorithm}[t]
  \caption{Reduction from \ec to \dsq}
  For each vertex $v \!\in\! V(H)$, create $r$ vertices $x^1_v, \dots, x^r_v$ and set $\floq(x^i_v) =1$ and $\fupq(x^i_v) = 2$ for each $i \in [r]$.
  
  For each edge $uv \in E(H)$ and each $i \in [r]$, create a vertex $y_{uv}^i$ with $\floq(y_{uv}^i)=1$, $\fupq(y_{uv}^i)=2$, and make it adjacent to $x_{u}^i$ and $x_{v}^i$.
  
  Create a vertex ${M}$ adjacent to all vertices created in the previous step:
  $\{y^i_{uv} : i \in [r], uv \in E(H)\}$, with $\floq({M})\!=\!\fupq({M})\!=\!1$.
  
  Create $2n\!+\!r\!+\!2$ pendants $\{p^M_j\}_{j=1}^{2n+r+2}$ adjacent to ${M}$ with $\floq(p^M_j)\!=\!\fupq(p^M_j)\!=\!1$ for each $j \!\in\! [2n\!+\!r\!+\!2]$.
    
  Let $l$ and $h$ denote the smallest and the largest sizes of the color classes: that is, $\ell = \lfloor n / r \rfloor$ and $h = \lceil n/r \rceil$.
  For each $i \in [r]$, create a vertex $c^i$ (to ``check'' the color class size) that is adjacent to all vertices $\{x_v^i\}_{v \in V(H)}$, with $\floq(c^i)  \!=\! \ell \!+\! 1$ and $\fupq(c^i) \!=\! h \!+\! 1$.

  For each vertex $c^i$, create $2n\!+\!r\!+\!2$ pendants $\{p^i_j\}_{j=1}^{2n+r+2}$ adjacent to $c^i$, with $\floq(p^i_j) \!=\! \fupq(p^i_j) \!=\! 1$.

  For each vertex $v \!\in\! V(H)$, create a vertex $\alpha_v$ (to ``force'' a vertex among its neighbors to be picked) adjacent to all vertices $\{x_v^i\}_{i \in [r]}$, with $\floq(\alpha_v) \!=\! \fupq(\alpha_v) \!=\! 2$.
  
  For each $\alpha_v$, create $2n+r+2$ pendants $\{p^v_j\}_{j=1}^{2n+r+2}$ adjacent to $\alpha_v$, with $\floq(p^v_j) = \fupq(p^v_j) = 1$.

  Set $k = 2n+r+1$.
  \label{reduction2}
\end{algorithm}
\begin{figure}[t]
  \centering
  \includegraphics[width=0.7\textwidth]{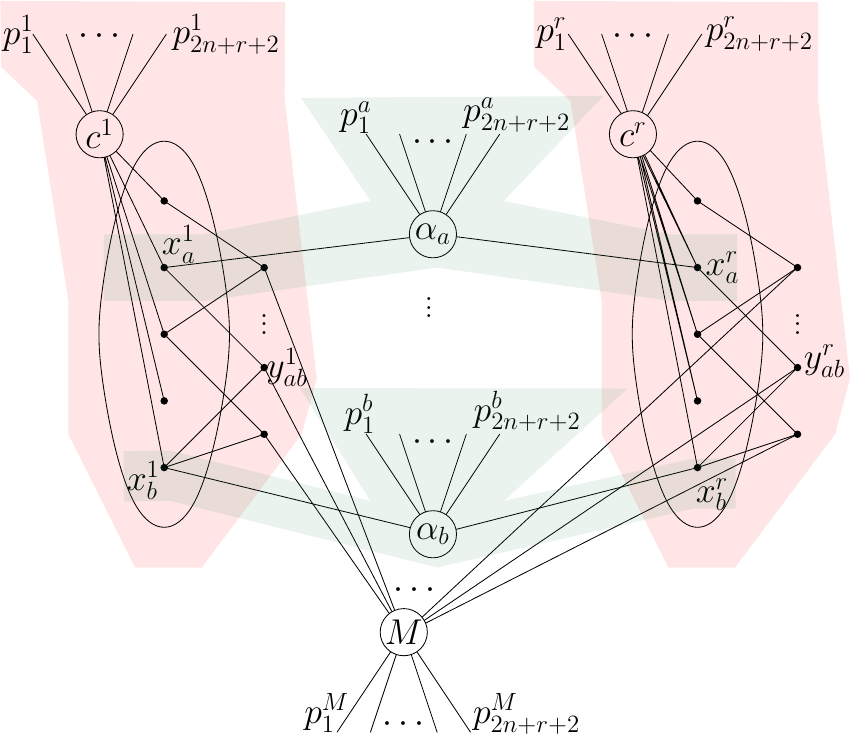}
  \caption{\label{fig:dsq-w-hardness-reduction-tw} The instance $\instJ$ constructed in the reduction from \ec. One can view the construction as $r$ verticals: one for each color and $n$ horizontals: one for each vertex.}
\end{figure}
Any \dsq of size at most $k$ contains $\{{M}\} \cup \{\alpha_v\}_{v \in V(H)} \cup \{c^i\}_{i \in [r]}$ since each of it is adjacent to $k+1$ pendants with quota $\langle 1,1 \rangle$.
Moreover, no solution contains those pendants.
From now on, we call those vertices \textit{newly added pendants}.
Next, we show the correctness of the reduction.

\begin{claim}
  $\instI$ is a yes-instance of $\ec$ if and only if $\instJ$ is a yes-instance of $\dsq$. 
\end{claim}
\begin{proof}
  Suppose that $\instI$ admits an equitable coloring $\mathrm{col} \colon V(H)\rightarrow[r]$.
  We claim that $S = \{{M}\} \cup \{\alpha_v\}_{v \in V(H)} \cup \{c^i\}_{i \in [r]} \cup \{x^i_v \colon v \in V(H), \mathrm{col}(v)=i\}$ is a \dsq in $\instJ$.
  All the newly added pendants are dominated properly.
  Each vertex $\alpha_v$, which has quota $\langle2,2\rangle$, is dominated properly: by itself, and by $x^{\mathrm{col}(v)}_v$.
  Any vertex of the form $x^i_v$, which has quota $\langle1,2\rangle$ is dominated by $c^i$ and possibly by itself (if $x^i_v \in S$).
  Since $\mathrm{col}$ is a proper vertex coloring, no two endpoints of an edge are colored the same.
  Thus, any vertex of the form $y^i_{uv}$, which has quota $\langle1,2\rangle$, is dominated by ${M}$ and by at most one of its neighbors: $x^i_{u}$ and $x^i_{v}$.
  Since $\mathrm{col}$ is an equitable coloring, each vertex $c^i$ is dominated by itself, and by between $\ell$ and $h$ vertices in $S \cap N(c^i) \!=\! \{x^i_v \colon v\in V(H), \mathrm{col}(v)\!=\!i\}$.
  Thus, $S$ is a $2n\!+\!r\!+\!1\!=\!k$ sized \dsq in $\instJ$.
  
  For the other direction, suppose that $S'$ is a \dsq of size at most $2n+r+1$ in $\instJ$.
  We have the following:
  \begin{itemize}[itemsep=0pt]
    \item $S'$ contains the set $\{{M}\} \cup \{\alpha_v\}_{v \in V(H)} \cup \{c^i\}_{i \in [r]}$ and none of the newly added pendants.
    \item For each vertex $\alpha_v \in V(G)$, there is a unique $i \in [r]$ such that $x^i_v \in S'$ (otherwise $\alpha_v$ is not dominated properly since its quota is $\langle 2,2 \rangle$ and $\alpha_v \in S'$).
    \item For each vertex of the form $y^i_{uv}$, at most one of $x^i_u$ and $x^i_v$ is in $S'$ since $x_M \in S'$.
    \item For each $i \in [r]$, we have $\ell \le |S' \cap \{x^i_v \colon v \in V(H)\}| \le h$.
  \end{itemize}
  Thus, we obtain $\mathrm{col} \colon V(H)\rightarrow [r]$, by mapping to each vertex $v\in V(H)$ the unique color $i\in [r]$ such that $x^i_v \in S'$, an equitable coloring using exactly $r$ colors in $\instJ$.
\end{proof}

To establish \W[1]-hardness parameterized by treewidth, we also show that $\tw(G)$ is bounded by a function of $\tw(H)$ and $r$, by exhibiting a tree decomposition of $G$ of such width.
\begin{claim}
    $\tw(G) \le (2r+2)\tw(H)$.
\end{claim}
\begin{proof}
\begin{sloppypar}
    Let $\mathcal{T}$ be a tree decomposition of $H$ with width $\tw(H)$.
    For each $v \in V(H)$, replace each occurence of $v$ in bags of $\mathcal{T}$ by vertices $\{c^1 \dots c^r, ~{M}, ~\alpha_v, ~x^1_v \dots x^r_v\}$.
  For each edge $uv \in E(H)$ and each $i \in [r]$, create a pendant node $q^i_{uv}$ in $\mathcal{T}$ with bag $\{x^i_u, y^i_{uv}, x^i_v\}$, adjacent to any node containing both $x^i_{u}$ and $x^i_v$ (such a node exists by the properties of $\mathcal{T}$).
  For each newly added pendant vertex $p$ with neighbor $\beta$, create a pendant node with bag $\{p,\beta\}$ adjacent to a node containing $\beta$. Note that since $\beta \in \{{M}\} \cup \{c^i\}_{i \in [r]} \cup \{\alpha_v\}_{v \in V(H)}$, such a node always exists by the first step in our construction.
Observe that any edge in $G$ is contained in a bag in $\mathcal{T}$; and the set of nodes containing any vertex $v \in V(G)$ form a nonempty connected subtree in $\mathcal{T}$.
Thus, $\mathcal{T}$ is a tree decomposition of $G$ of width $(2r+2) \tw(H)$.
Consequently, $\tw(G) \le (2r+2)\tw(H)$. \qedhere
\end{sloppypar}
\end{proof}


Overall, we have the following.
\begin{theorem}
  \label{thm:dsq-hard-tw}
  \dsq parameterized by treewidth is \textup{\W[1]}-hard.
\end{theorem}

\subsection{Algorithms}
\label{sec:dsq-bounded-treewidth-algorithms}
Using standard dynamic programming over tree decompositions, similar to that for \ds (see~\cite{cygan2015parameterized}), we first establish that \dsq is \fpt parameterized by $\tw + k$.
\begin{theorem}
  \label{thm:dsq-fpt-ktw}
  Given a tree decomposition of $G$ of width $\tw$, \dsq can be solved in time $2^\tw (k+1)^{2\tw} \cdot n^{\bigoh(1)}$.
\end{theorem}
\begin{proof}
  Towards streamlining dynamic programming algorithms on graphs of bounded treewidth, a tree decomposition $(T, \{X_t\}_{t \in V(T)})$ is defined to be \textit{nice} if it satisfies the following conditions:
  \begin{itemize}[itemsep=0pt]
    \item $T$ is rooted at a vertex $r$ and the bags corresponding to the root and the leaves are empty.
    \item each non-leaf node is of one of the following types:
    \begin{enumerate}[itemsep=0pt]
      \item \textbf{Introduce vertex node}: a node $t$ with exactly one child $t'$ such that $X_t = X_{t'} \cup \{v\}$ for some $v \notin X_{t'}$. The vertex $v$ is \textit{introduced} at $t$.
      \item \textbf{Introduce edge node}: a node $t$ labelled with an edge $uv \in E(G)$ such that $\{u,v\}\subseteq X_t$, with exactly one child $t'$ such that $X_t = X_{t'}$. The edge $uv$ is \textit{introduced} at $t$. Every edge of $E(G)$ is introduced exactly once in the whole decomposition.
      \item \textbf{Forget vertex node}: a node $t$ with exactly one child $t'$ such that $X_t = X_{t'} \setminus \{v\}$ for some $v \in X_{t'}$. The vertex $v$ is \textit{forgotten} at $t$.
      \item \textbf{Join node}: a node $t$ with exactly two children $t_1$ and $t_2$ such that $X_t = X_{t_1} = X_{t_2}$.
    \end{enumerate}
  \end{itemize}
  It is known that given a tree decomposition of $G$, a nice tree decomposition of $G$ of the same width can be computed in polynomial time \cite{cygan2015parameterized}.
  From now on, we will work on a nice tree decomposition $\mathcal{T} = (T,\{X_t\}_{t \in V(T)})$ of $G$ of width $\tw$.
  For a node $t \in V(T)$, let $V_t$ and $E_t$ denote the set of vertices and edges in the bags corresponding to the subtree rooted at $t$ and let $G_t$ denote the subgraph of $G$ containing vertices $V_t$ and edges $E_t$ (i.e., $G_t = G[V_t, E_t]$).
  
  For each node $t \in V(T)$, function $f\colon X_t \rightarrow [0,k]$ (to capture the number of times each vertex in $X_t$ is dominated by a hypothetical solution), $S \subseteq X_t$ (to capture the intersection of a solution with $X_t$), we compute $c[t,f,S]$: the cardinality of a minimum subset of vertices $D$ in $G_t$ such that $D$ dominates $V_t \setminus X_t$ properly (i.e., all vertices in the subtree $V_t$ but not in the bag $X_t$), $D \cap X_t = S$ and for each $v \in X_t$, $|D \cap N[v]|=f(v)$.
  Then the cardinality of a minimum \dsq in $G$ is $c[r,\emptyset, \{\}]$, where $\emptyset$ is the empty function.
  
  Next, we will see how the entries of the table $c$ are computed recursively.
  We will traverse the tree decomposition in a bottom-up manner.
  The leaves form the base case, and the correctness of the computation at a node $t$ will be proved using induction on the \textit{node height} (maximum distance between the node and a leaf).
  The entries are computed as follows depending on the type of $t$.
  
  \noindent$\bullet$ \textbf{Leaf node:}
  Since $X_t=\{\}$, we have $c[t,\emptyset,\{\}]=0$.
  
  Next, we will focus on computing an entry $c[t,f,S]$ for some $t \in V(T), f\colon X_t \rightarrow [0,k],$ and $S \subseteq X_t$. 
  For a subset $Y\subseteq X_t$, by $f|_Y$, we denote the restriction of $f$ to $Y$.
  For a vertex $v \in X_t$ and a value $i \in [0,k]$, by $f_{v \rightarrow i}$ we denote the function which takes the value $i$ at $v$ and the same value as $f$ at other points. That is, 
  \[
  f_{v \rightarrow i}(x) = 
  \begin{cases}
    i, & x=v;\\
    f(x), & x\ne v.
  \end{cases}
  \]
  
  \noindent$\bullet$ \textbf{Introduce vertex node:}
  Suppose that the vertex $v$ is introduced at $t$ and the child of $t$ is $t'$.
  Then,
  \begin{equation}
    \label{eq:int-ver-node}
    c[t,f,S] = 
      \begin{cases}
        \begin{alignedat}{3}
          c[t',f|_{X_{t'}}, S],                   &\quad v\notin S, f(v)=0;   \\
          \infty,                                 &\quad v\notin S, f(v)\ne0; \\
          \infty,                                 &\quad v\in S, f(v)\ne1;    \\
          1+c[t',f|_{X_{t'}}, S\setminus\{v\}],   &\quad v\in S, f(v)=1.
        \end{alignedat}
      \end{cases}
    \end{equation}
  The correctness follows from the fact that $v$ is an isolated vertex in $G_t$.
  
  \noindent$\bullet$ \textbf{Introduce edge node:}
  Suppose that the edge $uv$ is introduced at $t$ and the child of $t$ is $t'$.
  Then, 
  \begin{equation}
    \label{eq:int-edg-node}
    c[t,f,S] \!=\! 
      \begin{cases}
        \begin{alignedat}{2}
          c[t',f, S],                           &\quad u \notin S, v \notin S;                        \\
          \infty,                               &\quad u \in S,    v \notin S, f(v) = 0;            \\
          c[t',f_{v\rightarrow f(v)-1}, S],     &\quad u \in S,    v \notin S, f(v) \ne 0;          \\
          \infty,                               &\quad u \notin S, v \in S, f(u) = 0;            \\
          c[t',f_{u\rightarrow f(u)-1}, S],     &\quad u \notin S, v \in S, f(u) \ne 0;          \\
          \infty,                               &\quad u \in S,    v \in S, \min(f(u),\!f(v)) \! = \! 0; \\
          \!c[t'\!,\!f_{u\rightarrow f(u)-1, v\rightarrow f(v)-1}, S],     &\quad u \in S, v \in S, \min(f(u),\!f(v)) \! \ne \! 0.
            \end{alignedat}
      \end{cases}
  \end{equation}
  The correctness follows from the fact that the introduction of $uv$ at $t$ only affects the domination of $u$ and $v$.

\noindent$\bullet$ \textbf{Forget vertex node:}
Suppose that the vertex $v$ is forgotten at node $t$ and $t'$ is the child of $t$.
Then,
\begin{equation}
  \label{eq:for-node}
  c[t,f,S] \! = \!
  \min\limits_{i \in [\floq(v),\fupq(v)]} 
    \min \left( 
      c[t',f_{v \rightarrow i},S], c[t',f_{v \rightarrow i}, S \cup \{v\}] 
  \right)
\end{equation}
Any solution corresponding to $c[t,f,S]$ dominates $v$ exactly $i$ times for some $i \in [\floq(v),\fupq(v)]$.
Then any solution corresponding to either $c[t',f_{v \rightarrow i},S]$ or $c[t',f_{v \rightarrow i}, S \cup \{v\}]$ forms a feasible solution corresponding to $c[t,f,S]$ when $i \in [\floq(v),\fupq(v)]$.
Thus, we infer LHS $\le$ RHS in \Cref{eq:for-node}.
For the other direction, let $D$ be any solution corresponding to $c[t,f,S]$. 
Suppose that $D$ dominates $v$ exactly $i$ times for some $i \in [\floq(v),\fupq(v)]$.
If $D$ contains $v$, then it forms a feasible solution corresponding to $c[t',f_{v \rightarrow i}, S \cup \{v\}]$.
Otherwise, it forms a feasible solution corresponding to $c[t',f_{v \rightarrow i},S]$.
Then we have, $c[t,f,S] \ge \min(c[t',f_{v \rightarrow i},S], c[t',f_{v \rightarrow i}, S \cup \{v\}])$.
Consequently we infer LHS $\ge$ RHS in \Cref{eq:for-node}. 

\noindent$\bullet$ \textbf{Join node:}
Suppose that $t$ is a join node with children $t_1$ and $t_2$. 
Two mappings $f_1$ and $f_2$ are said to be \textit{compatible} (comp., in short) with $f$ if for each $v \in X_t$,
$f(v)=f_1(v)+f_2(v)-|\{v\} \cap S|$.
Then,
\begin{align}
  \label{eq:join-node}
  c[t,f,S] = \min\limits_{\text{comp. mappings }f_1,f_2} \left( c[t_1,f_1,S]+c[t_2,f_2,S]-|S| \right).
\end{align}
Consider any solutions $D_1,D_2$ corresponding to $c[t_1,f_1,S]$ and $c[t_2,f_2,S]$ resp.
Since $f_1$ and $f_2$ are comp., $D_1 \cup D_2$ is a feasible solution corresponding to $c[t,f,S]$.
Any vertex $v \notin X_t$ is dominated properly by $D_1 \cup D_2$.
Any vertex $v \in X_t$ is dominated exactly $f_1(v)+f_2(v) -|\{v\} \cap S|= f(v)$ times by $D_1 \cup D_2$.
Thus, we have LHS $\le$ RHS in \Cref{eq:join-node}.
For the other direction, consider any solution corresponding to $c[t,f,S]$.
Let $D_1 = D \cap V_{t_1}$ and $D_2 = D \cap V_{t_2}$.
By the property of $D$, any vertex $v \in V_{t_1} \setminus X_{t_1}$ is dominated properly by $D_1$ and let $f_1 \colon X_{t_1} \rightarrow [0,k]$ be the function that captures how many times each vertex in $X_{t_1}$ is dominated by $D_1$.
Let $D_2$ and $f_2 \colon X_{t_2} \rightarrow [0,k]$ be defined analogously.
For each $v \in X_t$, $v$ is dominated exactly $f_1(v)+f_2(v)-|\{v\} \cap S|$ times.
Thus,
\begin{align*}
\begin{aligned}
  c[t,f,S]  &= |D| = |D_1|+|D_2|-|D_1 \cap D_2| \\
  &= |D_1|+|D_2|-|S| \\
  &\ge c[t_1,f_1,S]+c[t_2,f_2,S]-|S| \\
  &\ge \min_{\text{comp. mappings }f_1,f_2} \left( c[t_1,f_1,S]+c[t_2,f_2,S]-|S| \right).
\end{aligned}
\end{align*}
  Given a function $f:X_t\rightarrow [0,k]$, the number of pairs of compatible mappings $f_1, f_2$ is
  bounded by $(k+1)^{|X_t|} \le (k+1)^{\tw+1}$.
  Thus, the time it takes to compute $c[t,f,S]$ for a given $t,f,S$ is bounded by $(k+1)^{\tw} \cdot n^{\bigoh(1)}$.
  
  The number of mappings $f$ is $(k+1)^{\tw}$.
  The number of subsets $S$ corresponding to a bag in $T$ is at most $2^{\tw}$.
  The total number of entries in table $c$ is $n^{\bigoh(1)} \cdot (k+1)^{\tw} \cdot 2^{\tw}$.
  Thus, the overall running time of the DP algorithm is $n^{\bigoh(1)} \cdot (k+1)^{2\tw} \cdot 2^{\tw} = 2^{\bigoh(\tw \log k)} \cdot n^{\bigoh(1)}$ and we have our desired result.
\end{proof}

Using a ``subset convolution'' like method in the join node, we can get an improvement in our algorithm. 
Let $\hat{u} = \max_{v \in V(G)} \fupq(v)$, the maximum upper-quota of a vertex. 
  Since we are interested in finding a solution of size at most $k$, we can assume that $\hat{u} \le k$ without loss of generality (each vertex can be dominated at most $k$ times).
  We would like to note that if we change the set of states in a node $t$ in the DP table by considering the functions $f:X_t \rightarrow [\hat{u}]$ instead of $f: X_t \rightarrow [k]$, we get an algorithm for \dsq that runs in time $2^{\tw}(\hat{u}+1)^{2\tw} \cdot n^{\bigoh(1)}$.
  If $\hat{u}$ is a fixed constant, then we can use the techniques of \cite{van2021generic} to obtain an algorithm for computing all the entries in the join node that runs in total time $2^{\tw}(\hat{u}+1)^{\tw} \cdot n^{\bigoh(1)} = (2\hat{u}+2)^{\tw} \cdot n^{\bigoh(1)}$.
  If $\hat{u}$ is not a fixed constant, then the techniques of \cite{van2021generic} run in time $((2\hat{u}+2)^{\tw} + 2^{\bigoh(\hat{u})}) \cdot n^{\bigoh(1)}$ (see Remark 4.20 in \cite{focke2022tight}).
  We discuss this next.

\subsubsection{On computing the entries at join nodes}
\label{subsec:computing-join-node-entries}
Consider the algorithm for \dsq given in the proof of \Cref{thm:dsq-fpt-ktw}.
Suppose that $t$ is a join node with children $t_1$ and $t_2$. 
For each function $g:X_t \rightarrow [0,\hat{u}]$ (to capture the number of times each vertex in $X_t$ is dominated by the forgotten vertices in a hypothetical solution (i.e., vertices in the solution other than S)), integer $q \in [0,n]$ (to capture the size of the solution), $S \subseteq X_t$ (to capture the intersection of a solution with $X_t$) we will compute $c'[t,g,S,q]$: it will be greater than zero if there is a solution $D$ for $G_t$ such that (1) each vertex $v \in X_t$ is dominated $g(v)$ times by $D \setminus X_t$, (2) $D \cap X_t = S$, and (3) $|D| = q$.
By computing the entries of the table $c'$, we also obtain the entries of the table $c$: $c[t,f,S]$ is the minimum $q$ such that $c'[t,g,S,q]>0$, where $g(v) = f(v)-|N_{G_t}[v] \cap S|$ for each $v \in X_t$.

Two mappings $g_1$ and $g_2$ are said to be \textit{compatible} (comp., in short) with $g$ if for each $v \in X_t$,
$g(v)=g_1(v)+g_2(v)$;
as a shorthand, we write $g=g_1+g_2$.
Consider the functions $f_1: X_{t_1} \rightarrow [0,\hat{u}]$ and $f_2: X_{t_2} \rightarrow [0,\hat{u}]$ where $f_1(v) = g_1(v)+|N_{G_{t_1}}[v] \cap S|$ and similarly $f_2(v) = g_2(v)+|N_{G_{t_2}}[v] \cap S|$ for each $v \in X_t$.
For an integer $q_1$, let $c'[t,g_1,S,q_1]= 1$ if $c[t,f_1,S] = q_1$ and $0$ otherwise. 
Similarly, for an integer $q_2$, let $c'[t,g_2,S,q_2]= 1$ if $c[t,f_2,S] = q_2$ and $0$ otherwise. 

Then,
\begin{align}
  \label{eq:new-join-node}
  c'[t,g,S,q] \!=\! \sum\limits_{\substack{g_1,g_2 \\ g=g_1+g_2}} \sum\limits_{\substack{q_1,q_2 \in [0,n] \\ q_1+q_2=q+|S|}} c'[t_1,g_1,S,q_1] ~ c'[t_2,g_2,S,q_2].
\end{align}
For the correctness, we claim that $c'[t,g,S,q] >0$ if and only if there is a solution $D$ for $G_t$ such that (1) each vertex $v \in X_t$ is dominated $g(v)$ times by $D \setminus X_t$, (2) $D \cap X_t = S$, and (3) $|D| = q$.
For the forward direction, consider any $g_1, g_2, q_1, q_2$ such that $c'[t_1,g_1,S,q_1] = c'[t_2,g_2,S,q_2] = 1$ and $g_1, g_2$ are comp. with $g$ and $q_1+q_2=q+|S|$ (their existence is guaranteed because $c'[t,g,S,q]>0$).
Let $D_1, D_2$ be the solutions corresponding to $c'[t_1, g_1, S, q_1]$ and $c'[t_2,g_2,S,q_2]$ respectively.
Since $g_1$ and $g_2$ are comp. and $q_1+q_2=q+|S|$, $D_1 \cup D_2$ is a feasible solution corresponding to $c'[t,g,S,q]$.
Any vertex $v \notin X_t$ is dominated properly by $D_1 \cup D_2$.
Any vertex $v \in X_t$ is dominated exactly $g_1(v)+g_2(v) = g(v)$ times by $(D_1 \cup D_2) \setminus X_t$.

For the other direction, suppose that $c[t,g,S,q]>0$ and consider any solution corresponding to $c[t,g,S,q]$.
Let $D_1 = D \cap V_{t_1}$ and $D_2 = D \cap V_{t_2}$.
By the property of $D$, any vertex $v \in V_1 \setminus X_{t_1}$ is dominated properly by $D_1$ and let $g_1:X_{t_1} \rightarrow [0,\hat{u}]$ be the function that captures how many times each vertex in $X_{t_1}$ is dominated by $D_1 \setminus S$.
Let $D_2$ and $g_2:X_{t_2} \rightarrow [0,\hat{u}]$ be defined analogously.
For each $v \in X_t$, we have $g_1(v)+g_2(v)=g(v)$.
Since $D_1 \cap D_2 = S$, we have $q_1+q_2 = q+|S|$.
Thus, we have that $c'[t_1, g_1, S, q_1] = c'[t_2, g_2, S, q_2] = 1$.
Consequently, $c'[t,g,S,q]>0$.

For each $w \in [0, \hat{u}(\tw+1)]$ (to capture the total number of times vertices in $X_t$ are dominated by forgotten vertices), we define $c'[t,g,S,q,w] = 0$ if $w \ne \sum_{v \in X_t} g(v)$ and otherwise
\small
\begin{align}
  \label{eq:newer-join-node}
  \begin{split}
  c'[t,g,S,q,w] \!=\!\!\!\!\!\! \sum_{\substack{g_1, g_2 \\ g \equiv g_1+g_2}} \sum_{\substack{q_1,q_2 \in [0,n] \\ q_1+q_2=q+|S|}} \sum_{\substack{w_1,w_2 \in [0,\hat{u}(\tw+1)] \\ w_1+w_2=w}} 
  \!\!\!\left(c'[t_1,g_1,S,q_1,w_1] ~ c'[t_2,g_2,S,q_2,w_2]\right),
  \end{split}
\end{align}
\normalsize
where by $g \equiv g_1+g_2$ we mean that the two sides are equal coordinate-wise modulo $(\hat{u}+1)$.
We claim that $c'[t,g,S,q,w] = c'[t,g,S,q]$ if $w = \sum_{v \in X_t} g(v)$.
It is sufficient to show that the RHS of \Cref{eq:new-join-node} and \Cref{eq:newer-join-node} are equal if $w = \sum_{v \in X_t} g(v)$.
Since $g=g_1+g_2$ implies that $g \equiv g_1+g_2$, we have that the RHS of \Cref{eq:new-join-node} is at most the RHS of \Cref{eq:newer-join-node}.
For the other direction, if the RHS of \Cref{eq:new-join-node} is strictly lesser than the RHS of \Cref{eq:newer-join-node}, then there is some $g_1, g_2$ such that $g \equiv g_1+g_2$ but $g \ne g_1+g_2$.
But since $g, g_1, g_2$ are functions from $X_t$ to $[0,\hat{u}]$, this implies that $w = \sum_{v \in X_t} g(v) < \sum_{v \in X_{t_1}} g_1(v) + \sum_{v \in X_{t_2}} g_2(v) = w_1+w_2$, a contradiction since $w=w_1+w_2$.

To compute $c'[t,g,S,q,w]$, we use the result of \cite{van2021generic}.
For a fixed choice of $t,S,q,w$, the table entries corresponding to all functions $g$ can be computed in time $((\hat{u}+1)^\tw+ 2^{\bigoh(\hat{u})})\cdot n^{\bigoh(1)}$ \cite{van2021generic,focke2022tight}.
Thus all the entries at a join node can be computed in time $((2\hat{u}+2)^\tw+ 2^{\bigoh(\hat{u})})\cdot n^{\bigoh(1)}$ (the last term $2^{\bigoh(\hat{u})}$ is for computing an appropriate prime and this step can be assumed to occur once in the entire algorithm; see Remark 4.20 in \cite{focke2022tight}).

\section{\dsq on graphs of bounded degeneracy}
\label{sec:simpler-dsq-3deg-graphs}

The degeneracy of a graph $G$ is the smallest number $d$ such that any subgraph of $G$ has a vertex of degree at most $d$.
Note that the degeneracy is bounded by the maximum degree.
In this section, we show the W[1]-hardness of \dsq even in constant degeneracy graphs by a parameterized reduction from the \textsc{Independent Set} problem (IS): given a graph $H$ and an integer $r$, does there exist a set $S \subseteq V(H)$ such that $S$ is an $r$-sized independent set in $H$ (i.e., $H[S]$ is edgeless)?


The \is problem parameterized by $r$ is W[1]-hard \cite{cygan2015parameterized}. 
The basic idea of the following reduction is to express the fact that any independent set contains at most one endpoint of an edge as a domination upper-quota.
Given an instance $\instI=(H,r)$ of $\is$, we construct an instance $\instJ=(G,k,\floq,\fupq)$ of \dsq in polynomial time as given in \Cref{reduction1}. 
See \Cref{fig:simpler-dsq-w-hardness-reduction-degen} for an illustration.
We would like to mention that the basic construction in the reduction is similar to that in a reduction presented in \cite{meybodi2020parameterized}, which establishes the hardness of a different related problem called $[1,j]$-\textsc{Dominating Set} on graphs of degeneracy $j+1$.

\begin{algorithm}[t]
  \caption{Reduction from \is to \dsq}
  For each vertex $v\in V(H)$, create a vertex $x_v$ in $G$ and set $\floq(x_v)\!=\!0$ and $\fupq(x_v)\!=\!1$.

  For each edge $uv \in E(H)$, create a vertex $y_{uv}$ with $\floq(y_{uv}) = 0$, $\fupq(y_{uv})=1$, and make it adjacent to $x_u$ and $x_v$.

  Create a vertex $\alpha$ (to ``force'' $r$ vertices in $H$ to be picked) that is adjacent to all of $\{x_v\}_{v \in V(H)}$, with $\floq(\alpha)=\fupq(\alpha)=r$.
  
  Create a pendant vertex $p^{\alpha}$ adjacent to $\alpha$, with $\floq(p^{\alpha}) \!=\! \fupq(p^{\alpha}) \!=\! 0$.
  
  Set $k = r$.
  \label{reduction1}
\end{algorithm}
\begin{figure}[t]
  \centering
  \includegraphics[width=0.8\textwidth]{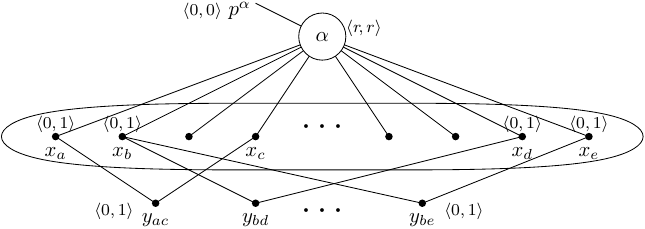}
  \caption{ The instance $\instJ$ constructed in the reduction from \is}
  \label{fig:simpler-dsq-w-hardness-reduction-degen}
\end{figure}

\begin{lemma}
  $\instI$ is a yes-instance of \is if and only if $\instJ$ is a yes-instance of \dsq.
\end{lemma}
\begin{proof}
  Suppose that $S=\{v_i\}_{i \in [r]}$ is an independent set in $\instI$. 
  We claim that $S' = \{x_{v_i}\}_{i \in [r]}$ is a solution in $\instJ$.
  We prove it by showing that each vertex in $V(G)$ is properly dominated by $S'$.
  \begin{itemize}[itemsep=0pt,leftmargin=*]
    \item[$\bullet$] The pendant $p^{\alpha}$ (with quota $\langle0,0\rangle$) is not dominated.
    \item[$\bullet$] The vertex $\alpha$ (with quota $\langle r, r \rangle$) is dominated by $r$ vertices: by $\{x_{v_i}\}_{i \in [r]}$.
    \item[$\bullet$] Consider a vertex $x_v \in V(G)$ (with quota $\langle0,1\rangle$). 
    If $v \in S$ then $x_v$ is dominated once: by itself; otherwise $x_v$ is not dominated.
    \item[$\bullet$] Since $S$ is an independent set in $H$, we have that $|\{u,v\} \cap S| \le 1$ for each $uv \in E(H)$.
    Thus, for each $uv \in E(H)$, the vertex $y_{uv}$ is dominated at most once by either $x_u$ or $x_v$.
  \end{itemize}
  
  For the other direction, suppose that $S'$ is a solution in $\instJ$.
  Firstly, observe that $S'$ does not contain $\alpha$: otherwise the pendant $p^\alpha$ (with quota $\langle 0,0 \rangle$) is dominated once.
  Since $\alpha$ has quota $\langle r,r \rangle$, we have that $S'$ contains exactly $r$ vertices from $\{x_v\}_{v \in V(H)}$.
  Let $S = S' \setminus \{\alpha\}$.
  Since each vertex $y_{uv} \in V(G)$ has quota $\langle0,1\rangle$, we have that at most one of $x_u$ or $x_v$ is in $S$.
  Therefore, the vertices in $H$ corresponding to $S$ form an independent set of size $k=r$ in $\instI$.
\end{proof}

The following ordering of vertices (in which each vertex has degree at most $2$ into the vertices succeeding it) shows that $G$ has degeneracy $2$: $p^{\alpha}$, $\{y_{uv}\}_{uv \in E(H)}$, $\{x_v\}_{v \in V(H)}$, $\alpha$.
Overall, we have the following.
\begin{theorem}
  \label{thm:dsq-hard-3dgen}
  \dsq parameterized by $k$ is \textup{W[1]}-hard on graphs of degeneracy 2.
\end{theorem}

\section{\dsq on nowhere dense graphs}
\label{sec:dsq-nd}

In this section, we will show that \dsq is $\fpt(k)$ if the input graph belongs to a graph class that is \textit{nowhere dense}.
We briefly introduce the required definitions for our result.
For a more extensive presentation, we refer to the paper \cite{grohe2017deciding}, which contains the algorithmic meta-theorem of our interest.
A \textit{model} of a graph $H$ in $G$ is a function $\psi$ that maps vertices in $H$ to vertex disjoint connected subgraphs in $G$ (called \textit{branch sets}) such that $\psi(u)$ and $\psi(v)$ are adjacent if $uv \in E(H)$.
It is known that $H$ is a minor of $G$ if and only if there is a model of $H$ in $G$.
The graph $H$ is said to be an \textit{$r$-shallow minor} of $G$ if there is a model $\psi$ of $H$ in $G$ such that for each $v \in H$, the radius of $\psi(v)$ (that is, the minimum maximum distance from a vertex, given by $\min_{x \in \psi(v)}\max_{y \in \psi(v)} \mathrm{dist}(x,y)$) is at most $r$.
A class $\C$ of graphs is said to be \textit{nowhere dense} if there is a function $f\colon \mathbb{N} \rightarrow \mathbb{N}$ such that for all $r$, $K_{f(r)}$ is not an $r$-shallow minor of $G$ for all $G \in \C$.
We focus on the case where $f$ is computable: such classes are called \textit{effectively} nowhere dense.
Note that this is a reasonable restriction since the functions corresponding to all natural nowhere dense graph classes are computable \cite{grohe2017deciding}.
Next, we argue that \dsq is $\fpt(k)$ if the input graph is nowhere dense.
On a high level, we construct a new graph that belongs to a nowhere dense graph class and ``contains'' $G$ along with the quota information of each vertex, and use first-order model-checking on this graph to determine the existence of a solution.

\paragraph{FO model-checking.}
We follow the notation and terminology of Grohe et al.~\cite{grohe2017deciding}. 
A \textit{vocabulary} $\sigma$ is a finite set of relation symbols, each with a fixed arity.
A $\sigma$-structure consists of a set $A$ called the \textit{universe} and an interpretation of relations in $\sigma$ in $A$.
The interpretation is also referred to as coloring.
The \textit{atomic} formulas of $\sigma$ are of the form $x=y$ (checking equality) or $R(x_1, \dots, x_r)$, where $R \in \sigma$ is an $r$-ary relation and $x_1, \dots, x_r$ are variables.
The \textit{first-order} (FO, in short) formulas of $\sigma$, denoted by $\mathrm{FO}[\sigma]$, are obtained from atomic formulas using boolean connectives $\lor,\land,\lnot$ and quantifiers $\exists$ (existential), $\forall$ (universal).
A sentence $\varphi$ is a formula in which all variables are quantified.
Graphs may be viewed as ${E}$-structures, where $E$ is a binary relation symbol that captures the edges.
A colored-graph vocabulary consists of the binary relation symbol $E$ and possibly finitely many unary relation symbols. 
A $\sigma$-colored graph is a $\sigma$-structure whose ${E}$-restriction is a simple undirected graph.
In the FO \textit{model-checking} problem on a class ${\C}$ of graphs, we are given a graph ${G} \in {\C}$ and an FO-sentence $\varphi$. 
The goal is to determine whether ${G} \models \varphi$ (i.e., ${G}$ satisfies $\varphi$).
The following result by Grohe et al. \cite{grohe2017deciding}, which considers formulas with a free variable $x$ (i.e., $x$ is not quantified), implies that FO model-checking on ${\C}$ is $\fpt$ in formula size if ${\C}$ is nowhere dense.

\begin{proposition}[Theorem 8.1 in \cite{grohe2017deciding}]
  Let $\C$ be an effectively nowhere dense class of graphs. There is a computable function $f$ and an algorithm that, given $\epsilon >0$, a formula $\varphi(x) \in \mathrm{FO}(\sigma)$ for some colored-graph vocabulary $\sigma$ and a $\sigma$-colored graph $G \in \C$, computes the set of all $v \in V(G)$ such that $G \models \varphi(v)$ in time $f(|\varphi|, \epsilon) \cdot n^{1+\epsilon}$.

\end{proposition}

Suppose that the input graph $G$ belongs to a nowhere dense graph class $\C$ such that each graph in $\C$ excludes $K_{f(r)}$ as an $r$-shallow minor for each $r \in \mathbb{N}$.
In particular, it excludes $K_{f(0)}$ as a subgraph (which is equivalently a 0-shallow minor).
The basic idea is to construct another nowhere dense graph $G'$ in which the quota information of each vertex is embedded in the graph structure itself; and then work on $G'$.
One can consider adding $2(k+1)$ different unary relations to the vocabulary for expressing the lower and upper-quotas of the vertices in the graph and then writing an FO formula based on that to solve \dsq 
\footnote{As pointed out by an anonymous reviewer, it appears that this approach of adding new unary relations can be used to obtain the more general result that \dsq is $\fpt(k)$ in monadically stable graph classes, where FO model checking is $\fpt$ in the formula size due to the result by Dreier et al., \cite{DBLP:conf/focs/DreierEMMPT24}. Since we are unclear about how adding unary relations affects the overall running time of the algorithm, we have refrained from doing that.}. 
But this results in different vocabularies for different values of $k$.
We present an algorithm using only the edge relation in the vocabulary that decides whether $G$ has a \dsq of size at most $k$.

For each vertex $v$, we attach a \textit{quota gadget} to it: $(v)$ $-$ $(K_{f(0)})$ $-$ (path on $\floq(v)+1$ vertices) $-$ $(K_{f(0)})$ $-$  (path on $\fupq(v)+1$ vertices) $-$ $(K_{f(0)})$, where the hyphens denote an edge incident to an arbitrary vertex in the complete graph).
See \Cref{fig:quota-gadget} for an example.
We call a vertex a \textit{gadget vertex} if it is part of a quota gadget, and call it an \textit{original vertex} otherwise.
Let $G'$ be the graph obtained from a copy of $G$ by attaching the quota gadgets to each vertex.
Since $K_{f(0)}$ is not a subgraph of any graph in $\C$, we have that $G' \notin \C$.
We claim that $G'$ belongs to a different nowhere dense graph class.

\begin{figure}[t]
  \centering
  \includegraphics[width=0.5\textwidth]{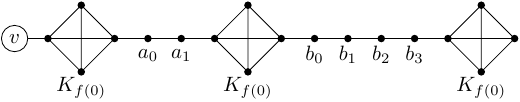}
  \caption{\label{fig:quota-gadget} The gadget for a vertex $v$ with quota $\langle 1,3 \rangle$ in an instance where $f(0)\!=\!4$}
\end{figure}

\begin{lemma}
  \label{lem:nowhere-dense}
  The graph $G'$ excludes $K_{g(r)}$ as an $r$-shallow minor for each $r \in \mathbb{N}$, where $g(r)=\max(f(r),3f(0)+2))$.
  Hence, it belongs to the nowhere dense graph class characterized by the function $g$.
\end{lemma}
\begin{proof}
  Fix $r \! \in \! \mathbb{N}$.
  Let $H$ be an $r$-shallow minor of $G$ with a model $\psi$.
  Suppose that $H$ is isomorphic to $K_\ell$ for some $\ell \! \in \! \mathbb{N}$.
  We will argue that $\ell \le \max\{f(r),3f(0)\!+\!2\}$.
  
  \begin{observation}
    \label{obs:gadget-path}
    Let $p_1, p_2, p_3$ be any three gadget vertices corresponding to $v \in V(G)$ that do not belong to the three $K_{f(0)}$'s (i.e., they are in the paths between them) in the gadget. 
    They do not belong to three distinct branch sets in $\psi$.
  \end{observation}
  \begin{proof}
    Assume without loss of generality that $p_1$, $p_2$, and $p_3$ occur in that order in the gadget corresponding to $v$.
    Suppose that for each $i \in [3]$, $p_i$ occurs in branch set $\psi(q_i)$ for some $q_1, q_2, q_3 \in V(H)$. 
    Observe that $p_2$ separates $p_1$ from $p_3$.
    By construction, there is no edge between $\psi(q_1)$ and $\psi(q_3)$ since each branch set is connected in $G'$.
    Then, the model $\psi$ does not correspond to a complete graph, which is a contradiction.
  \end{proof}
  
  \begin{observation}
    \label{obs:no-original}
    If no branch set in $\psi$ contains an original vertex, then $\ell \! \le  \! 3 f(0) \! + \! 2$.
  \end{observation}
  \begin{proof}
    Since the gadgets corresponding to two distinct original vertices are disjoint, and the model $\psi$ corresponds to a clique, we have that the gadget vertices in the branch sets all belong to one original vertex.
    The total number of vertices in a gadget corresponding to an original vertex $v$ is $3 f(0) + \floq(v) + \fupq(v)$.
    By Observation~\ref{obs:gadget-path}, no three path vertices are in distinct branch sets.
    Thus, we have $\ell \le 3 f(0)+2$.
  \end{proof}
  
  \begin{observation}
    \label{obs:orig-gadg-diff}
    Suppose there is an original vertex $v$ such that $v$ is contained in a branch set and a gadget vertex corresponding to $v$ is contained in another branch set, then $\ell \le 3f(0)+2$.
  \end{observation}
  \begin{proof}
    Since $v$ and a gadget vertex corresponding to $v$ are in two distinct branch sets and $\psi$ corresponds to a clique, we have that the vertices in the branch sets other than the one containing $v$, are the gadget vertices corresponding to $v$.
    Thus, similar to the proof of Observation~\ref{obs:no-original}, we have $\ell \le 3 f(0)+2$.
  \end{proof}
  The only case left to argue is when an original vertex and its gadget vertices do not appear in separate branch sets (possibly do not appear in any branch set). 
  \begin{observation}
    \label{obs:allinone}
    Suppose that for all original vertices $v$, 
    the number of distinct branch sets containing $v$ or any gadget vertex corresponding to $v$ is at most one.
    Then, $\ell \le f(r)$.
  \end{observation}
  \begin{proof}
    Suppose that an original vertex $v$ is present in a branch set $\psi(w)$ for some $w \in V(H)$.
    Let $Q_v$ denote the gadget vertices corresponding to $v$.
    Consider the model $\psi_v$ of $H$ in $G'$ where $\psi_v(x)=\psi(x)\setminus Q_v$ if $x=w$, and $\psi_v(x)=\psi(x)$ otherwise.
    Observe that $\psi_v(w)$ is connected.
    Let $x,y \in \psi_v(w)$.
    The shortest path between $x$ and $y$ in $G'$ does not contain a vertex in $Q_v$.
    Thus, the radius of $\psi_v(w)$ is at most the radius of $\psi(w)$, which is bounded by $r$ because $H$ is an $r$-shallow minor.
    Therefore, $\psi_v$ is a model of $H$ in $G'$ with each branch set having radius at most $r$.
    Repeat the above procedure for each vertex that satisfy the observation statement and obtain a model $\hat{\psi}$.
    Observe that $\hat{\psi}$ is also a model of $H$ in $G$ with each branch set having radius at most $r$.
    Since $G \in \C$, the minor corresponding to $\hat{\psi}$, namely $H$, is not $K_{f(r)}$.
    Thus, we have $\ell < f(r)$.
  \end{proof}
  Combining Observations \ref{obs:gadget-path} to \ref{obs:allinone}, we have $\ell \le \max(3f(0)+2, f(r))$.
\end{proof}

Next, we use FO model-checking on $G'$ to solve \dsq in $G$.
Here, the universe $A$ is the vertices of the constructed graph $G'$ and $\sigma$ contains the binary relation $E$ that captures the edges in $G'$.
Since we have shown that $G'$ belongs to a nowhere dense graph class, we are left with showing the first order expressibility of the problem.
By construction and properties of $G$ and $G'$, we have the following.

\begin{lemma}
    \label{lem:gadget-correctness}
  A vertex $v$ is present in $G$ with quota $\langle \ell,u \rangle$ if and only if in $G'$, $v$ is attached to the following: a clique on $f(0)$ vertices, a path on $\ell+1$ vertices, a clique on $f(0)$ vertices, a path on $u+1$ vertices, a clique on $f(0)$ vertices; in the order as listed with each part connected by an edge.
\end{lemma}
We show that there is a formula of length $\bigoh(2^k)$ on $G'$ that solves \dsq in $G$.
Consequently, we have that \dsq is $\fpt$ parameterized by $k$ if $G \! \in \! \C$ for a nowhere dense class $\C$.
Without loss of generality, we assume that $k>0$ and at least one vertex has non-zero lower-quota. If $k=0$, then the input is a yes-instance if and only if all vertices have lower-quota 0.
Moreover, if all vertices have lower-quota 0, then the input is a yes-instance.

\begin{itemize}
    \item 
  Formula \textbf{Dom}($u,v$) of length $\bigoh(1)$ that checks the domination of a vertex by another:
  \[\textbf{Dom}(u,v) := (u=v) \lor E(u,v).\]
  \item 
  Formula \textbf{Distinct}($x_1,\dots,x_s$) of length $\bigoh(s^2)$ that checks whether the input arguments are distinct:
  \[
  \textbf{Distinct}(x_1, \dots, x_s) := \land_{i \ne j \in [s]} \lnot (x_i = x_j).
  \]
  \item Formula \textbf{Clique}($q_1, \dots, q_r$) of length $\bigoh(r^2)$ that checks whether 
  $\{ q_1, \dots, q_r \}$ form a clique:
  \[
  \textbf{Clique}(q_1, \dots, q_r) := \land_{i\ne j \in [r]} E(q_i,q_j)
  \]
  \item Formula \textbf{Path}($q_1, \dots, q_r$) of length $\bigoh(r)$ that checks whether 
  $\langle q_1, \dots, q_r \rangle$ form a path:
  \[
  \textbf{Path}(q_1, \dots, q_r) := \land_{i \in [r-1]} E(q_i,q_{i+1})
  \]
  \item
  Formula \textbf{Check}{$_{\ell,u}$}($w$) of length $\bigoh(k^2)$ (note that $f(0)$ is a positive constant and we assume $\ell, u \in [0,k]$) that checks whether $w$ is an original vertex (equivalently, $w$ is not a gadget vertex) with $\floq(w)=\ell$ and $\fupq(w)=u$:
  \small
  \begin{equation*}
    \begin{aligned}
    \textbf{Check}_{\ell, u}(w) := 
    &\exists
    c^1_1, ..., c^1_{f(0)},
    a_0, ..., a_\ell, 
    c^2_1, ..., c^2_{f(0)}, 
    b_0, ..., b_u, 
    c^3_1, ..., c^3_{f(0)} \\
    &\textbf{Distinct}\Big(c^1_1, ..., c^1_{f(0)},
    a_0, ..., a_\ell, 
    c^2_1, ..., c^2_{f(0)}, 
    b_0, ..., b_u, 
    c^3_1, ..., c^3_{f(0)}\Big) \\
    &\land E(v,c^1_1) \land \textbf{Clique}(c^1_1,...,c^1_{f(0)}) \land E(c^1_{f(0)}, a_0) 
    \\& \land \textbf{Path}(a_0, ...,a_\ell) \\
    &\land E(a_\ell, c^2_1) \land \textbf{Clique}(c^2_1,...,c^2_{f(0)}) \land E(c^2_{f(0)}, b_0) \\& 
    \land \textbf{Path}(b_0, ...,b_u) \\
    &\land E(b_u, c^3_1) \land \textbf{Clique}(c^3_1,...,c^3_{f(0)}).
    \end{aligned}
  \end{equation*}
  \normalsize
  Observe that the correctness of the above formula follows from \Cref{lem:gadget-correctness}.
  \item 
  Formula \textbf{SatLQ}{$_\ell$}($v, x_1,\dots,x_s$) of length $\bigoh(2^k k)$ (we assume $s \le k$ and $\ell \in [0,k]$) that checks whether the vertex $v$ is dominated by at least $\ell$ vertices in $x_1, \dots, x_s$:
  
  If $\ell \in [s]$,
  \[
  \textbf{SatLQ}_{\ell}(v, x_1, \dots, x_s) := \bigvee\limits_{\substack{S \subseteq \{x_1, \dots, x_s\} \\ |S|=\ell}} \left(\bigwedge_{u \in S} \textbf{Dom}(u,v) \right)
  \]
  
  Else if $\ell >s$,
  \[
  \textbf{SatLQ}_{\ell}(v, x_1, \dots, x_s) := \lnot \left( v=v \right)
  \]
  
  Else $\ell=0$,
  \[
  \textbf{SatLQ}_{\ell}(v, x_1, \dots, x_s) := \left( v=v \right)
  \]
  Note that the formula always evaluates to True when $\ell=0$, and to False when $\ell > s$.
  \item
  Formula \textbf{SatUQ}$_{u}$($v,x_1,\dots,x_s$) of length $\bigoh(2^k k)$ (we assume $s\le k$ and $u \in [0,k]$) that checks whether the vertex $v$ is dominated by at most $u$ vertices in $x_1, \dots, x_s$:

  If $u\in [0,s-1]$,
  \begin{align*}
    \textbf{SatUQ}_{u}&(v, x_1, \dots, x_s) :=
    & \lnot \left(\bigvee_{\substack{S \subseteq \{x_1, \dots, x_s\} \\ |S|=u+1}} \left(\bigwedge_{u \in S} \textbf{Dom}(u,v) \right) \right)    
  \end{align*}
  
  Else $u\ge s$,
  \[
  \textbf{SatUQ}_{u}(v, x_1, \dots, x_s) := \left( v=v \right)
  \]
  Note that the formula always evaluates to True when $u \ge s$.
  \item
  Formula $\varphi_s$ of length $\bigoh(2^k k^3)$ (we assume $s\le k$) that combines everything and checks whether there is a solution of size exactly $s$:
    \small
  \begin{align*}
    &\varphi_s := \exists x_1,\dots,x_s \, \forall v \,
    \Bigg[
    \textbf{Distinct}(x_1, \dots ,x_s) \land
    \bigwedge_{i \in [s]}
    \Bigg(
    \bigvee_{\substack{\ell \in [0,k]\\u \in [\ell,k]}} \textbf{Check}_{\ell,u}(x_i)
    \Bigg)
    \Bigg]
    \\
    & \land
    \Bigg[
    \bigwedge_{\substack{\ell \in [0,k]\\u \in [0,k]}}
    \Bigg(
     \lnot \textbf{Check}_{\ell,u}(v) \lor
    \Big(
    \textbf{SatLQ}_{\ell}(v, x_1, \dots, x_s)  \land \textbf{SatUQ}_{u}(v, x_1, \dots, x_s)
    \Big)
    \Bigg)
    \Bigg]
\end{align*}
\normalsize
  Observe that the last part of the above formula is equivalent to checking whether $v$ is an original vertex with quota $\langle \ell,u \rangle$ and if that is the case, then $v$ is properly dominated by $x_1, \dots, x_k$.
  \item 
  Formula $\varphi$ of length $\bigoh(2^k k^4)$ that checks whether there is a solution of size at most $k$:
  \[
  \varphi := \bigvee_{s \in [k]} \varphi_s.
  \]
\end{itemize}

By construction, we have that $G' \models \varphi$ if and only if there is a \dsq of size $k$ in $\instI$.
Since the size of $\psi$ is bounded by a function of $k$, we have the following.
\begin{theorem}
  \label{thm:dsq-fpt-nd}
  \dsq parameterized by $k$ is \fpt on nowhere dense graphs.
\end{theorem}

\section{\dsq on apex-minor-free graphs}

  
  A graph is said to be \textit{planar} if it can be drawn on a plane with no two edges crossing each other.
  By Wagner's theorem, a graph is planar if and only if it does not contain $K_5$ or $K_{3,3}$ as a minor.
  A graph is an \textit{apex} graph if deleting a vertex results in a planar graph.
  A graph class $\C$ is said to be \textit{apex-minor-free} if all graphs in $\C$ avoid a fixed apex graph $H$ as a minor.
  Observe that planar graphs are also apex-minor-free.
  We show that \dsq is $\fpt(k)$ in apex-minor-free graphs using the \textit{bidimensionality} approach \cite{fomin2009contraction}.
  
  
  On a high level, the bidimensionality approach exploits the structure of graphs with large treewidth. 
  Towards that, we use the algorithm in \Cref{sec:dsq-bounded-treewidth-algorithms} for \dsq on graphs of bounded treewidth. 

In this section, we design a sub-exponential time \fpt algorithm for \dsq on apex-minor-free graphs. For that, we observe that, in polynomial time, we can discard some lower-quota 0 vertices depending on their neighborhoods.
\begin{rr}
  \label{rr:deg0}
  If there is a vertex $v$ with lower-quota 0 such that all vertices at a distance at most 2 from $v$ have lower-quota 0, then we delete $v$.
\end{rr}
To establish the correctness of the above, we show the following.
\begin{claim}
    Let $\instI$ and $\instJ$ denote the instances $(G,k,\floq,\fupq)$ and $(G\setminus \{v\}, k,$ 
    $\floq|_{V(G)\setminus\{v\}}, \fupq|_{V(G)\setminus\{v\}})$ respectively. Then, $\instI$ is a yes-instance if and only if $\instJ$ is a yes-instance.
\end{claim}
\begin{proof}
 Let $D$ be a solution in $\instI$.
 If $v \notin D$, then $D$ is also a solution in $\instJ$.
 Suppose that $v \in D$.
 Since any vertex dominated by $v$ has lower-quota 0, the set $D\setminus \{v\}$ is a solution in $\instJ$.
 For the other direction, let $D'$ be a solution in $\instJ$.
 Observe that any vertex dominated by $N^1_G(v)$ is at distance at most $2$ from $v$ in $G$ and thus has lower-quota $0$. 
 Hence, the set $D' \setminus N^1_G(v)$ is a solution in $\instI$.
\end{proof}

A vertex satisfying the above condition can be found in polynomial time. 
Thus, exhaustive application of the above reduction rule can also be done in polynomial time.
From now on, we will assume so: for any vertex $v \in V(G)$, there is a vertex $w \in N^2_G[v]$ such that $\floq(w) \ge 1$.
Any solution in $\instI$ contains a vertex $w_d$ that dominates $w$ and $w_d$ is at distance at most 3 from $v$.
Consequently, any solution in $\instI$ is also a \textit{3-dominating set} (3-DS, in short) in $G$: a set \hide{of vertices} $S$ such that any vertex $v \in V(G)$ is at distance at most $3$ from $S$.

\begin{sloppypar}
For an integer $q$, the $(q \times q)$-grid is a graph with vertices $\{(i,j)\}_{i,j \in [q]}$
and edges $\{((i_1,j_1),(i_2,j_2))\colon |i_1-i_2|+|j_1-j_2|=1\}$.
The graph $\Gamma_q$ is obtained from a $(q \times q)$-grid by triangulating all inner faces so that all internal vertices have degree 6 (say by adding the edges $\{((x+1,y),(x,y+1))\colon x,y \in [q-1]\}$) and then connecting one corner vertex of degree 2 with all vertices of the external face (say by adding the edges $\{((q,q),(x,y))\colon x\in\{1,q\} \text{ or } y\in\{1,q\}\}$).
For obtaining an algorithm for \dsq, we use the following result due to Fomin et al.
\end{sloppypar}
\begin{proposition}[Theorem 1 in \cite{fomin2009contraction}]
  \label{prop:cont-apex}
  Let $H$ be an apex graph. There is a constant $c_H$ such that every connected graph $G$ excluding $H$ as a minor and of treewidth at least $c_H \cdot q$, contains $\Gamma_q$ as a contraction.
\end{proposition}
Let $D_3^G$ denote a minimum 3-DS in $G$.
Observe that for any graph $C$ that is a contraction of $G$, we have $|D_3^G| \ge |D_3^C|$.
This is because contracting any edge cannot increase the size of a minimum 3-DS.
Now we are ready to present the algorithm for the case where 
 $G$ is $H$-minor-free for a fixed apex graph $H$.

It is known that there is a $2^{\bigoh(s)} n^{\bigoh(1)}$ algorithm that, given the graph $G$ and an integer $s$, either returns a tree decomposition of $G$ of width at most $4s+4$ or concludes that $\tw(G)>s$ \cite{cygan2015parameterized}.
We set $s=c_H \cdot 15\sqrt{k}$ and run the algorithm: we consider two cases depending on its output.
If it concludes that $\tw(G)>s$, then
by \Cref{prop:cont-apex}, $G$ contains $\Gamma_{15 \sqrt{k}}$ as a contraction.
Observe that any internal vertex can 3-dominate at most 49 vertices (in the 7x7 grid that it is the center of).
Thus, any \dsq in $G$, which is also a 3-DS, is of size at least $|D_3^{\Gamma_{15 \sqrt{k}}}| \ge (225k-(4\cdot3\cdot(15\sqrt{k}-1)))/49 > k$ (the last inequality holds for all $k\ge 1$; in the first inequality, we use an upperbound on the number of vertices that the corner vertex adjacent to all external vertices can 3-dominate in $\Gamma_{15 \sqrt{k}}$) and we return that there is no DSQ of size $k$.
In the other case, the algorithm returns a tree decomposition $\mathcal{T}$ of $G$.
The width of $\mathcal{T}$ is at most $4 \cdot c_H \cdot 15 \sqrt{k}+4 = \bigoh(\sqrt{k})$ and we run the DP algorithm (in \Cref{sec:dsq-bounded-treewidth-algorithms}) on $\mathcal{T}$ and obtain an answer in time $2^{\bigoh(\sqrt{k} \log k)}\cdot n^{\bigoh(1)}$.
Thus, we have the following.

\begin{theorem}
  \label{thm:dsq-fpt-amf}
  \dsq can be solved in time $2^{\bigoh(\!\sqrt{k} \log k)} \!\cdot\! n^{\bigoh(1)}\!$ on apex-minor-free graphs.
\end{theorem}

\section{\scqfull when all sets are \textit{small}}
\label{sec:appendix-sc}
The \textsc{Set Cover} problem (\textsc{SC}, in short) is a generalization of \ds in which we are given a universe $\U$, a family $\F \subseteq 2^{\U}$, and an integer $k$, and the goal is to find a subfamily $S \subseteq \F$ of size at most $k$ such that $\bigcup_{R \in S} R = \U$. While it is $W[2]$-hard parameterized by $k$, it is \fpt w.r.t. $k$ and $d$ when the sets are of size at most $d$ \cite{DBLP:journals/siamcomp/DowneyF95,cygan2015parameterized}. 
 We define \scqfull (\scq, in short) as follows. Given a universe $\U$, a family $\F \subseteq 2^{\U}$, coverage quotas $\floq,\fupq\colon \U\rightarrow \mathbb{N}\cup\{ 0\}$, and an integer $k$, determine whether there is a subfamily $S\subseteq \F$ of size at most $k$ such that for each $u\in \U$ we have $\floq(u) \le \vert \{R\colon R \in S \text{ and } u \in R\}\vert \le \fupq(u)$.
Two natural restrictions that we study are the following: 
\begin{enumerate*}
  \item[(1)] \scq when each element occurs in at most $f$ sets (bounded frequency);
  \item[(2)] \scq when each set contains at most $d$ elements (bounded size).
\end{enumerate*}
\begin{observation}
  \label{obs:scq-comp}
  \scq is \textup{NP}-hard even if $f\!=\!2$; and is \fpt par. by $(k,f)$.
\end{observation}
\begin{proof}
  (i.) In the \textsc{Vertex Cover} problem, we are given a graph $H=(V,E)$ and an integer $k$.
  The goal is to find a vertex subset $S$ of size at most $k$ such that for each edge $uv \in E(H)$, either $u \in S$ or $v \in S$.
  It is \textsf{NP}-hard.
  It can be reduced to \scq with $f=2$: the universe is the set of edges, and for each vertex the family contains the set of edges covered by it (each edge appears in the sets corresponding to its endpoints).
  
  (ii.) Observe that \scq parameterized by $k+f$ is \fpt by a simple branching algorithm: for any element $e$ with $\floq(e) \ge 1$, branch on choosing the sets that contain it.
  The correctness follows from the fact that at least one of those sets must be in any solution.
\end{proof}

In this section, we will consider instances of \scq where each set in the set family has size at most $d$.
We first present a randomized algorithm based on the \textit{color-coding} technique \cite{alon1995color,cygan2015parameterized}. 
Existing results on derandomizations of the technique lead to deterministic algorithms for the problem we are interested in.

We call an element $u\in\U$ \textit{interesting} if $\floq(u)\ge 1$, and \textit{uninteresting} otherwise.
Suppose that the input admits a solution $S=\{R_i\}_{i=1}^{k}$. 
Observe that the number of elements covered by any $k$ sized solution (even when counting multiplicity) is at most $kd$.
This gives an upper-bound on the number of interesting elements (the number of uninteresting elements is unbounded though).
The idea is to randomly color the universe with $kd$ colors so that the elements covered by $S$ are colorful (i.e., no two covered elements 
have the same color) and then determine the existence of such a solution based on the coloring.

Let $C = (C_1, \dots, C_{kd})$ be a random $kd$-partition of $\U$; such a partition/coloring can be obtained by independently performing the following for each element $u\in \U$: choose $i\in[kd]$ uniformly at random and place $u$ in $C_i$. 
From now on, we will condition on the event that the elements covered by $S$ are colorful according to $C$: this happens with probability at least $\binom{kd}{\ell}\ell!/(kd)^{\ell} \ge e^{-\ell} \ge e^{-kd}$, where $\ell= |\bigcup_{R \in S} R| \le kd$.
Then, any part $C_i$ contains at most one interesting element: if there are two interesting elements in $C_i$, then since both are covered by $S$, the coloring $C$ does not satisfy the event condition). 
If a part $C_i$ contains an interesting element $u$, then delete from $\mathcal{U}$ all elements $e\ne u \in C_i$ and the sets that contain $e$.
This deletion is safe since the unique element in $C_i$ covered by $S$ is $u$ (i.e., $C_i \cap \bigcup_{R \in S} R = \{u\}$).
Thus, we have the following.

\begin{observation}
\label{obs:structure-of-parts}
  Any part $C_i$ is either empty, or is of the form $\{e_i\}$ for some interesting element $e_i$, or is a set of uninteresting elements.
\end{observation}

Let $\F=\{E_1, \dots, E_m\}$.
W.l.o.g., we may assume that each $E_i$ is colorful: otherwise we can delete $E_i$ since the elements covered by picking such a set will certainly not be colorful.
To look for a solution, we will compute the entries of the table $T$ defined as follows.
For $j \!\in\! [m]$ and a $kd$-tuple $X \!\in\! [0,k]^{kd}$, let $T[j,X]$ be the size of a minimum size subfamily in ${E_1, \dots, E_j}$ that covers elements in $C_i$ exactly $X[i]$ times (counting multiplicity).
We compute the entries of $T$ as follows. For the base case,
\begin{equation}
  \label{eq:smallsets-dp-basecase}
  T[0,X] = 
  \begin{cases}
    0 \text{\quad if $X=(0, \dots,0)$,}\\
    \infty \text{\quad otherwise.}
  \end{cases}
\end{equation}  
For each $j\in [m]$ and $X \in [0,k]^{kd}$ such that $\sum_{i \in [kd]} X[i] \le kd$, we compute $T[j,X]$ as follows.
Let $\mathrm{colormul}(E_j)$ (defined to capture the multiplicities of colors in $E_j$) denote an $kd$-length integer vector where the $i^{th}$ entry is the number of elements from color class $C_i$ in $E_j$.
\\
If $|E_j \cap C_i|>X[i]$, then 
\begin{equation}
  \label{eq:smallsets-dp-badset1-rec}
  T[j,X] = T[j-1,X].
\end{equation}
If there is an uninteresting element $u\in E_j$ such that $u \in C_i$ and $\fupq(u)<X[i]$, then 
\begin{equation}
  \label{eq:smallsets-dp-badset2-rec}
  T[j,X] = T[j-1,X].
\end{equation}
The condition $u \in C_i$ and $\fupq(u)<X[i]$ implies that $u \notin \bigcup_{R \in S} R$, the set of elements covered by the hypothetical solution.
Thus, the above equation is correct.
Otherwise if there are no uninteresting elements that satisfy the above condition, we have
\begin{equation}
  \label{eq:smallsets-dp-goodset-rec}
  T[j,X] = \min\left( T[j\!-\!1,X], 1\!+\!T[j\!-\!1,X\!-\!\mathrm{colormul}(E_j)]\right).
\end{equation}
\paragraph{Correctness of \Cref{eq:smallsets-dp-goodset-rec}:} Consider a solution corresponding to the entry $T[j-1,X]$. It gives an upperbound for $T[j,X]$ since it also satisfies the conditions of that entry. Similarly, the set $E_j$ along with a solution corresponding to the entry $T[j-1,X-\mathrm{colormul}(E_j)]$ gives an upperbound for $T[j,X]$. 
Thus, we have LHS $\le$ RHS in \Cref{eq:smallsets-dp-goodset-rec}. 
For the other direction, consider a solution $S$ corresponding to the entry $T[j,X]$.
If $E_j \in S$, then the removal of $E_j$ gives the upperbound $T[j-1,X-\mathrm{colormul}(E_j)] \le T[j,X]-1$.
Otherwise, we have the upperbound $T[j-1,X] \le T[j,X]$ since $S$ satisfies the conditions of the entry.
Thus, we have LHS $\ge$ RHS in \Cref{eq:smallsets-dp-goodset-rec}.

For each $j \in [m]$ and $X \in [0,k]^{kd}$ such that $\sum_{i \in [kd]} X[i] \le kd$,
the entries $T[j,X]$ are computed according to \Cref{eq:smallsets-dp-basecase,eq:smallsets-dp-badset1-rec,eq:smallsets-dp-badset2-rec,eq:smallsets-dp-goodset-rec}.
The number of $kd$-tuples $X\in [0,k]^{kd}$ such that $\sum_{i \in [kd]} X[i] \le kd$ is at most  $kd\binom{kd+k-1}{k-1} = 2^{\bigoh(k \log d)}$.
Thus, the entries of $T$, which can be enumerated in a counter-like fashion starting from $0^{kd}$, can be computed in time $2^{\bigoh(k \log d)} (|\U|+|\F|)^{\bigoh(1)}$.
To compute the answer for the input, it suffices to check if there exists $X$ such that $T[m,X] \!\le\! k$ and $\floq(e) \!\le\! X[i] \!\le\! \fupq(e)$ for each part $C_i$ that contains a unique interesting element $e$.
Note that all uninteresting elements have lower-quota 0 and \Cref{eq:smallsets-dp-badset1-rec,eq:smallsets-dp-badset2-rec} ensure that their upper-quotas are not violated.
Recall that we conditioned on the event that the elements covered by a hypothetical solution $S$ are colorful according to the random coloring $C$.
The event happened with probability at least $e^{-kd}$.
To obtain a deterministic algorithm from this, we use a standard pseudorandom object called \textit{perfect hash family}: a family of functions $\mathcal{H}\colon \U \rightarrow [kd]$ such that for each subset $Q \subseteq \U$ of size $kd$ there is a function $f_Q \in \mathcal{H}$ that is injective on $Q$.
Observe that each function in the family $\mathcal{H}$ can be viewed as a coloring and repeating the dynamic programming for all colorings in $\mathcal{H}$ gives us a deterministic algorithm.
By the properties of $\mathcal{H}$, if the input admits a solution $S$ of size $\le k$ (that covers at most $kd$ elements), then there is a coloring $C \in \mathcal{H}$ according to which $\bigcup_{R \in S} R$ is colorful and thus the DP procedure will find it.
It is known that one can construct such a perfect family $\mathcal{H}$ of size $e^{kd} (kd)^{\bigoh(\log kd)} \log (|\U|)$ in time linear in the size \cite{naor1995splitters}.
Thus, we have a deterministic algorithm for \scq that runs in time $2^{\bigoh(k d)}(|\U|+|\F|)^{\bigoh(1)}$.
\begin{theorem}
  \label{thm:scq-small}
  \scq with sets of size at most $d$ can be solved in time $2^{\bigoh(kd)} \cdot (|\U|+|\F|)^{\bigoh(1)}$.
\end{theorem}

\section{Conclusion and Open Questions}
Our results leave several intriguing questions open for future investigation. 
Firstly, our positive result for nowhere dense graphs relies on a reduction to FO model-checking. A natural direction for future work is to develop more direct and combinatorial \fpt{} algorithms for \dsq{} on these graph classes, avoiding the reliance on logical machinery. It would also be interesting to explore whether similar logic-based techniques can be extended to characterize even broader classes of graphs that admit efficient algorithms for \dsq{}.

Another promising direction concerns the approximability of \dsq{}, which remains largely unexplored. Can one design \fpt{}-approximation algorithms for \dsq{}, especially in settings where exact solutions are unlikely due to strong hardness results?

Finally, the question of preprocessing is open. In particular, can we obtain efficient kernelization algorithms for \dsq{}, at least on restricted graph classes? Given the success of kernelization techniques for \ds{} on sparse graphs, extending these ideas to \dsq{} could yield significant algorithmic insights.





\bibliographystyle{splncs04}
\bibliography{bibliography}



\end{document}